\documentclass[journal]{IEEEtran}

\usepackage[utf8]{inputenc}
\usepackage[T1]{fontenc}
\usepackage{url}
\usepackage{ifthen}
\usepackage{cite}
\usepackage[cmex10]{amsmath} 

\usepackage{mypackage}
\usepackage{mathrsfs}

\newtheorem{remark}{Remark}
\usepackage{mathtools}

\theoremstyle{approximation}

\newtheorem{proposition}{Proposition}

\usepackage{physics}
\usepackage{siunitx}

\newcommand{\sir}{\mathrm{SIR}}
\newcommand{\Pb}{\mathbb{P}}
\newcommand{\Eb}{\mathbb{E}}

\newcommand{\Lc}{\mathcal{L}}

\newcommand{\Eq}[1]{(\ref{eq:#1})}
\newcommand{\App}[1]{Appendix~\ref{app:#1}}
\newcommand{\Fig}[1]{Fig.~\ref{fig:#1}}

\usepackage{multicol}

\usepackage[nomain,acronym,shortcuts]{glossaries}
\makeglossaries
\newcommand*{\acro}[3][]{\newacronym[#1]{#2}{#2}{#3}}

\acro{D2D}{device-to-device}
\acro{SIR}{signal-to-interference-ratio}
\acro{SINR}{signal-to-interference-plus-noise-ratio}
\acro{PCP}{Poisson cluster process}
\acro{CoMP}{coordinated multi-point}
\acro{BS}{base station} 
\acro{MD-CoMP}{macrodiversity CoMP transmission}
\acro{MAC}{medium-access-control}
\acro{JT-CoMP}{joint transmission CoMP}
\acro{CoMP-JT}{coordinated multipoint joint transmission}
\acro{SBS}{small base station}
\acro{MDSD}{multiple devices to the single device}
\acro{MDS}{maximum distance seperable}
\acro{SCN}{small cell network}
\acro{PPP}{Poisson point process}
\acro{TCP}{Thomas cluster process}
\acro{CSI}{channel state information}
\acro{PDF}{probability distribution function}
\acro{PMF}{probability mass function}
\acro{RV}{random variable}
\acro{i.i.d.}{independently and identically distributed}
\acro{MBMS}{multimedia broadcasting multicasting service}
\acro{EE}{energy efficiency}
\acro{HCP}{hard-core placement}
\acro{CCDF}{complementary cumulative distribution function}
\acro{CDF}{cumulative distribution function}
\acro{PC}{probabilistic caching}
\acro{RC}{random uniform caching}
\acro{CPF}{caching popular files} 
\acro{PGFL}{probability generating functional}
\acro{KKT}{Karush-Kuhn-Tucker}
\acro{PGF}{point generating function}

\begin{document}
\title{Cooperative Transmission and Probabilistic Caching for Clustered D2D Networks}

\author{\IEEEauthorblockN{Ramy Amer\IEEEauthorrefmark{1},			
Hesham ElSawy\IEEEauthorrefmark{2},
Jacek Kibi\l{}da\IEEEauthorrefmark{1},
M. Majid Butt\IEEEauthorrefmark{1}\IEEEauthorrefmark{3},
Nicola Marchetti\IEEEauthorrefmark{1}}\\
\IEEEauthorblockA{\IEEEauthorrefmark{1}CONNECT, Trinity College, University of Dublin, Ireland}\\
\IEEEauthorblockA{\IEEEauthorrefmark{2}King Fahd University of Petroleum and Minerals (KFUPM)}\\
\IEEEauthorblockA{\IEEEauthorrefmark{3}Nokia Bell Labs, France}

\IEEEauthorblockA{email:\{ramyr, majid.butt, kibildj, nicola.marchetti\}@tcd.ie, hesham.elsawy@kfupm.edu.sa}}

\maketitle
\begin{abstract}
In this paper, we aim at maximizing the cache offloading gain for a clustered \ac{D2D} caching network by exploiting probabilistic caching and cooperative transmission among the cluster devices. Devices with surplus memory probabilistically cache a content from a known library. A requested content is either brought from the device's local cache, cooperatively transmitted from catering devices, or downloaded from the macro base station as a last resort. Using stochastic geometry, we derive a closed-form expression for the offloading gain and formulate the offloading maximization problem. In order to simplify the objective function and obtain analytically tractable expressions, we derive a lower bound on the offloading gain, for which a suboptimal solution is obtained when considering a special case. 
 Results reveal that the obtained suboptimal solution can achieve up to $12\%$ increase in the offloading gain compared to the Zipf's caching technique. Besides, we show that the spatial scaling parameters of the network, e.g., density of clusters and distance between devices in the same cluster, play a crucial role in identifying the tradeoff between the content diversity gain and the cooperative transmission gain.
\end{abstract}
\begin{IEEEkeywords}
\ac{D2D} communication, caching, clustered-process, CoMP. 
\end{IEEEkeywords}
\section{Introduction}
The deployment of low power \acp{BS} such as micro-, pico-, and femto-\acp{BS} provide short-range communication links and results in a higher density of spatial reuse of radio resources and thus in higher overall network throughput. However, deploying such a dense heterogeneous network comes with its own challenges. One such a challenge is the deployment cost associated with connecting all the small cells to the backbone with fast links. Motivated by this, caching finite popular files at mobile devices or access points in advance is considered a promising technique to relieve the overloaded network traffic.  We are here particularly interested in device caching and \ac{D2D} communication. 

The architecture of device caching exploits the large storage available in modern smartphones to cache multimedia files that might be highly demanded by the devices. Devices can exchange multimedia content stored in their local storage with nearby devices \cite{ji2016fundamental, amer2017delay,amer2018inter,amer2018optimizing}. Since the distance between a requesting device and a device that stores the requested file, called catering device, is small in most cases, \ac{D2D} communication is commonly used for content transmission \cite{ji2016fundamental}. As more than one device might cache the same content, the \ac{SINR} can be improved by joint transmission of the same cached content, which is denoted as cooperative transmission, e.g., via \ac{CoMP} transmission. 

The application of wireless caching along with \ac{CoMP} transmission, where \acp{BS} (or devices) cooperatively serve a content, is widely adopted in literature \cite{7565183,8125681,ao2015distributed,zheng2017optimization,amer2018sky}. For example, a cache-aided transmission scheme for downlink \ac{CoMP} with limited local cache resources at the \ac{BS} is proposed in\cite{8125681} to improve the outage performance. The content is divided into two sets, a popular set and a less popular set. Based on this heterogeneity in file popularity, a cache placement strategy that minimizes the average outage probability is derived. In \cite{ao2015distributed}, the authors propose to combine distributed caching of content in small cells and cooperative transmissions from nearby \acp{BS} to achieve content delivery speeds while reducing backhaul cost and delay. In particular, it is reported that the optimal caching strategy tends to either cache different content to maximize the hit ratio, or cache the same content, albeit in multiple \acp{BS} to achieve multiplexing gains. 
In general, caching at the \acp{BS} or mobile devices with \ac{CoMP} transmission reveals a tradeoff between content diversity gain, when caching and serving diverse content, and cooperative transmission gain, i.e., joint transmissions of the same content from multiple devices/caches \cite{chae2017content,chen2016cooperative,daghal2018content}.		

Compared with the above existing works, the scope of the current paper is to investigate and maximize the cache offloading gain for a clustered \ac{D2D} caching network. The devices are spatially clustered according to a \ac{TCP}, and have unused memory to cache some files following a random probabilistic caching scheme. 
Our network model effectively captures the stochastic nature of channel fading and the clustering nature of devices, which is not addressed yet in the literature, particularly in the context of caching and \ac{CoMP} transmission. \textcolor{black}{We formulate the offloading gain maximization problem, and a lower bound is then obtained, which is tractable. Based on the obtained lower bound, a  closed-form suboptimal caching solution is obtained for the offloading maximization problem. Simulation results show considerable improvement in the offloading gain as compared to other benchmark caching techniques.}

The rest of this paper is organized as follows. Section II and Section III present the system model and the offloading gain characterization, respectively. The rate analysis is introduced in Section IV and the suboptimal caching probability is obtained in Section V. Numerical results are then presented in Section VI, and Section VII concludes the paper.
\section{System Model} 
We consider a wireless caching network, where the mobile devices are randomly deployed and jointly transmit
their cached content to serve a device. The set of locations of the devices $\Phi$ is modelled as a \ac{TCP}, wherein the parent points are drawn from a homogeneous \ac{PPP} $\Phi_p$ with density $\lambda_p$, and the offspring points are \ac{i.i.d.} around each parent point \cite{daley2007introduction}. We will refer to the parent points and offspring as cluster centers and cluster members, respectively. Following \cite[Definition 3.5]{haenggi2012stochastic}, the locations of cluster members around a cluster center at $x\in \Phi_p$, $x \in \mathbb{R}^2$, are sampled from a normal distribution with variance $\sigma^2 \in \mathbb{R}^2$ forming a Gaussian \ac{PPP}, denoted as $\Phi_c$.  
Therefore, the density function of the location of a cluster member relative to its cluster center is
\begin{equation}
f_Y(y) = \frac{1}{2\pi\sigma^2}\textrm{exp}\Big(-\frac{\lVert y\rVert^2}{2\sigma^2}\Big),	
\quad\quad		y \in \mathbb{R}^2
\label{pcp}
\end{equation}
where $\lVert \cdot \rVert$ denotes the Euclidean norm. Finally, if $\overline{n}$ denotes the mean number of members per cluster, the intensity of the process $\Phi$ is $\lambda=\overline{n}\lambda_p$ while the intensity of Gaussian \ac{PPP} $\Phi_c$ is given by 
$\lambda_c(y) = \overline{n}f_Y(y)$.

We consider out-of-band \ac{D2D} communication whereby there is no cross-interference between the cellular network and \ac{D2D} communication. All devices are equipped with a single transmit-receive isotropic antenna and they have no \ac{CSI} from the device they are sending their content to. Furthermore, each D2D transmission uses all the available bandwidth, and the transmitted signals experience single-slope path loss with attenuation exponent $\alpha > 2$ and the power fading, which we model as \ac{i.i.d.} complex Gaussian \ac{RV} with zero mean and unit variance. 
\subsection{Content Popularity and Probabilistic Caching Placement}
We assume that each device has a surplus memory of size $M$ files designated for caching content from a known file library $\mathcal{F}$. The total number of files is $N_f> M$ and the set of content indices is denoted as $\mathcal{F} = \{1, 2, \dots , N_f\}$. These $N_f$ files represent the content catalog that all the devices in a cell may request, which are indexed in a descending order of popularity. The probability that the $m$-th file is requested follows a Zipf distribution given by 
\begin{equation}
q_m = \Bigg(m^{\beta} \sum_{k=1}^{N_f}k^{-\beta}\Bigg)^{-1},
\label{zipf}
\end{equation}
where $q_m$ represents the probability of having the $m$-th file requested, and $\beta$ is a parameter reflecting how skewed the popularity distribution is. Indeed, the lower indexed content has higher popularity, and by definition, $\sum_{m=1}^{N_f}q_m = 1$. It is also assumed that content of interest to the devices might be different across clusters. Therefore, we use Zipf distribution to model the popularity of files per cluster.

We adopt a random content placement where each file $m$ is cached independently at each device according to the probability $c_m$, such that $0 \leq c_m \leq 1$ for all $m=\{1, \dots, N_f\}$. To avoid duplicate caching of the same file within the memory of the same device, we follow the caching approach proposed in \cite{geographic_caching}, which requires that $\sum_{m=1}^{N_f}c_m=M$. Notice that devices caching content $m$ in a given cluster can be modeled as an Gaussian \ac{PPP}  $\Phi_{cm}$ with the intensity function given by the independent thinning theorem as $\lambda_{cm}(y)=c_m \lambda_c(y)$ \cite{haenggi2012stochastic}.
\section{Maximum Offloading Gain}
In this section, we present the proposed probabilistic caching system with cooperative transmission.
We adopt \ac{D2D} communication between cluster members to share files among each other. We also introduce the offloading gain as our key performance metric, which is the probability that a device obtains a desired content either from the local cache or via \ac{D2D} communication, at a rate that is at least $\rho$ \SI{}{bits/sec/Hz}. 

Our cooperative caching scheme works as follows. If a device requires the $m$-th file, first, it searches for the file in its own internal memory. If the requested file is cached in the internal memory, no \ac{D2D} link is scheduled for that file. However, if the requested file is not cached in the device's memory, the file can be downloaded via \ac{CoMP} transmission from all neighboring devices that cache the file in the same cluster, henceforth called catering devices. If the content is not cached entirety in its own cluster, the requesting device requests a file download from the network via the nearest \ac{BS}. 

Given stationarity of the parent process and independence of the offspring process, we can conduct our analysis for the representative cluster, which is an arbitrary cluster whose center is located at $x_0\in \Phi_p$, and the typical device, which is a randomly selected member of the representative cluster and requests the content. Without loss of generality, we assume the typical device is located in the origin $(0, 0) \in \mathbb{R}^2$. When the catering devices jointly transmit the same content $m$, the signal received at the typical device consists of two main components: the desired signal, which is the joint non-coherent transmission from the catering devices in the representative cluster, and the interference, which is created by other active \ac{D2D} links in clusters other than the representative cluster, henceforth called remote clusters. This can be formally denoted as
\begin{align}
y_m =&\underbrace{\sum_{y_0 \in \Phi_{cm}}\sqrt{\gamma_d} G_{y_0} \lVert x_0 + y_0\rVert^{-\alpha/2}  s_{y_0}}_{\text{desired signal}} \nonumber \\
&+  \underbrace{\sum_{x \in \Phi_p^{!}} \sum_{y \in \Phi_{ca}} \sqrt{\gamma_d} G_{y} \lVert x + y\rVert^{-\alpha/2} s_y}_{\text{interference}} + z,
\end{align}				 
where $y_0 \in \Phi_{cm}$ represents the set of catering devices for content $m$,  $G_{y_0}$ denotes the power fading between a catering device at $y_0 \in \Phi_{cm}$ relative to its cluster center at $x_0$ and the typical device, see Fig.~\ref{distance_near}, $\gamma_d$ denotes the \ac{D2D} transmission power, and $s_{y_0}$ is the symbol jointly transmitted by the catering devices $y_0 \in \Phi_{cm}$. 
$\Phi_p^{!}= \Phi_p \setminus\{x_0\}$ denotes the set of remote clusters, $s_y$ is the symbol transmitted by device in $y$, and $\Phi_{ca} \subseteq \Phi_{c}$ represents the set of transmitting devices in a remote cluster centered around $x\in\Phi_p^{!}$. 
$z$ denotes the standard additive white Gaussian noise and $G_{y}$ denotes the power fading between a potential interfering device at $y$ relative to its cluster center at $x$ and the typical device, see Fig.~\ref{distance_near}.

We focus on the interference-limited scenario with the effect of the thermal noise disregarded. Assuming unit power Gaussian symbols and treating interference as noise, the received \ac{SIR} at the typical receiver when downloading content $m$ is given by
\begin{align}
\sir_{m} = \frac{\gamma_d \Bigg|\sum_{y_0 \in \Phi_{cm}}G_{y_0} \lVert x_0 + y_0\rVert^{-\alpha/2}\Bigg|^2}{I_m}
\end{align}
where $I_m$ is the sum of interfering signal power associated with the downloading of content $m$, given by
\begin{align}
I_m = \sum_{x \in \Phi_p^{!}} \sum_{y \in \Phi_{ca}} \gamma_d G_{y} \lVert x + y\rVert^{-\alpha}
\end{align}
Accordingly, the instantaneous achievable data rate when downloading content $m$ cooperatively from all the catering devices can be defined as
\begin{align}
R_{m} = \log_2\big[1+ \sir_{m}\big],
\end{align}
in \SI{}{bits/sec/Hz}. \textcolor{black}{It is assumed that \ac{CoMP} transmission is adopted among all clusters, i.e., $\Phi_{ca}$ represents the set of active transmitters in the cluster centered at $x$. To simplify the mathematical analysis, we consider the worst case interference scenario when all the remote clusters' devices transmit a certain file, i.e., when $\Phi_{ca} \to \Phi_{c}$}. For notational simplicity, we henceforth drop the superscript $m$ from the interference. 

\begin{figure} [!t] 
\centering
\includegraphics[width=0.28\textwidth]{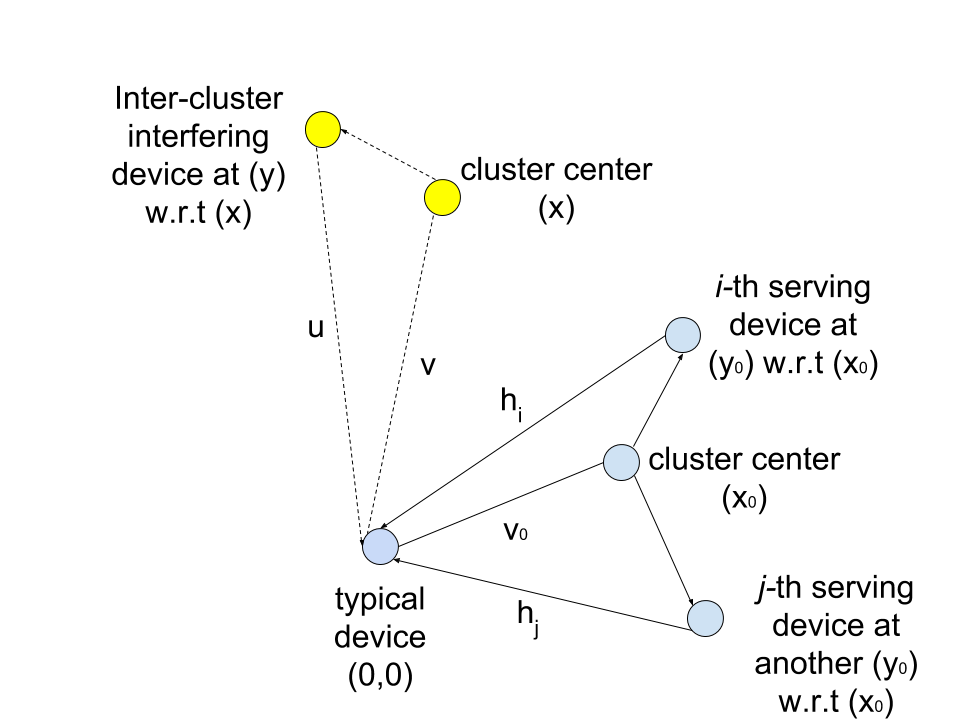}		
\caption {Illustration of the representative cluster and one interfering cluster.}
\label{distance_near}
\end{figure}

Now, we are ready to introduce an expression for the offloading gain at the typical receiver as 
\begin{align}
\label{eq:offload}
\mathbb{P}_{o}(\boldsymbol{c}) =  \sum_{m=1}^{N_f} q_m\Big(c_m &+ (1-c_m)\mathbb{P}(R_{m}>\rho) \Big)
 \end{align}
 
\section{Rate Analysis}
\textcolor{black}{In this section, we conduct the achievable rate analysis to obtain the probability $\mathbb{P}(R_{m}>\rho$).}
The probability that the achievable rate when downloading the content $m$ via \ac{CoMP} transmission is larger than $\rho$ is equal to
 \begin{align}
\Pb(R_{m}>\rho) &= \mathbb{P}\big(\log_2\big[1+ \sir_{m}\big] >\rho \big)			
\nonumber \\
&= \Pb\big(\sir_{m} >2^\rho-1 \big)			
\nonumber \\
&= \Pb\Bigg[ \frac{ \gamma_d \Big|\sum_{y_0 \in \Phi_{cm}}G_{y_0} \lVert x_0 + y_0\rVert^{-\alpha/2}\Big|^2}{I} \geq \theta\Bigg],
 \end{align}
where $\theta=2^\rho-1$. 
Since $G_{y_0}$ are \ac{i.i.d.} and $\sim \mathcal{CN}(0,1)$, we have $\Big|\sum_{y_0 \in \Phi_{cm}}\lVert x_0 + y_0\rVert^{-\alpha/2} G_{y_0}\Big|^2 \sim \exp\Big(\frac{1}{\sum_{y_0 \in \Phi_{cm}}\lVert x_0 + y_0\rVert^{-\alpha}}\Big)$. Then we have     
\begin{align}
\Pb(R_{m}>\rho) &\overset{}{=}  \Eb\Big[\exp\Big(-\frac{\theta I}{\gamma_d S_{\Phi_{cm}}}\Big)\Big] \nonumber \\
&\overset{(a)}{=}  \Eb\left[ \Lc_{I} \Big(t=\frac{\theta}{\gamma_d s_m}\Big)\Bigg|S_{\Phi_{cm}}=s_m\right]
\label{eq:rp-exact}
\end{align} 
where $S_{\Phi_{cm}}=\sum_{y_0 \in \Phi_{cm}}\lVert x_0 + y_0\rVert^{-\alpha}$ is a random variable that can be physically interpreted as the signal power from the catering devices $y_0 \in \Phi_{cm}$ subject to path-attenuation only (as we already averaged over the fading), assuming normalized power, and (a) is the Laplace transform of the inter-cluster interference evaluated at $t=\frac{\theta}{\gamma_d s_m}$.

In the ensuing subsections, we first derive the Laplace transform of the inter-cluster interference to obtain $\Pb(R_{m}>\rho)$. Then, we propose an upper bound on the interference to simplify the calculation of $\Pb(R_{m}>\rho)$, and correspondingly, the offloading gain.
The Laplace transform of the inter-cluster interference is presented in the following Lemma.
\begin{lemma}
\label{ch4:comp-interference}				
The Laplace transform of the inter-cluster interference, conditioned on the realization of the transmitters in the representative cluster $\Phi_{cm}$  serving content $m$ is expressed as
\begin{align}
\label{eq:lap-transform}
\Lc_{I}(t) = {\rm exp}\Big(-2\pi \lambda_p\int_{v=0}^{\infty}\Big(1 - e^{-\overline{n}\zeta(v,t)}\Big)v\dd{v}\Big)
\end{align}
where $t=\frac{\theta}{\gamma_d s_m}$, $\zeta(v,t) = \int_{u=0}^{\infty}\frac{t\gamma_d}{u^{\alpha}+t\gamma_d} f_{U|V}(u|v)\dd{u}$, $f_{U|V}(u|v)=\mathrm{Rice} (u;v,\sigma)$ is the Rician \ac{PDF} modeling the distance $U=\lVert x+y\rVert$ between an interfering device at $y$ relative to its cluster center at $x \in \Phi_{p}$ and the origin $(0,0)$ conditioned on $V=\lVert x\rVert=v$.
\end{lemma}
\begin{IEEEproof}
Please refer to \App{proof-comp-interference}. 
\end{IEEEproof}
We continue on by characterizing the serving distances' distribution. For a given realization $S_{\Phi_{cm}}=s_m$, let us assume that there are $k$ catering devices in the representative cluster. 
Let us also denote joint distances from the typical device (origin) to the $k$ serving devices in the representative cluster $x_0$ as $\boldsymbol{H}_k= \{H_1, \dots, H_k\}$. Then, conditioning on $\boldsymbol{H}_k = \boldsymbol{h}_k$, where $\boldsymbol{h}_k= \{h_1, \dots, h_k\}$, the conditional \ac{PDF} of the joint serving distances' distribution is denoted as $f_{\boldsymbol{H}_k}(\boldsymbol{h}_k)$.
%

 Since a serving device $i$ in the representative cluster $x_0$ has its coordinates in $\mathbb{R}^2$ chosen independently from a Gaussian distribution with standard deviation $\sigma$, then, by definition, the distance from such a serving device to the origin, denoted as $h_i=\lVert x_0+y_0\rVert$, $y_0\in \Phi_{cm}$, has Rician distribution $f_{H_i|V_0}(h_i|v_0)=\mathrm{Rice}(h_i;v_0,\sigma)$ conditioned on $V_0=\lVert x_0\rVert=v_0$. 
Since also the serving devices and the typical device have their locations sampled from a normal distribution with variance $\sigma^2$ relative to  their cluster center $x_0$, then, by definition, the statistical distance distribution between any two points, e.g., from the $i$-th serving device to the typical device, follows Rayleigh distribution $f_{H_i}(h_i)=\mathrm{Rayleigh}(h_i,\sqrt{2}\sigma)$. 

If the serving distances from the typical device to the different points of the cluster were independent from each other, $f_{\boldsymbol{H}_k}(\boldsymbol{h}_k)$ would simply be the product of $k$ independent \acp{PDF}, each of them given by $f_{H_i}(h_i)=\mathrm{Rayleigh}(h_i,\sqrt{2}\sigma)$. 
However, there is a correlation between the serving distances due to the common factor $x_0$ in the serving distance equation $h_i=\lVert x_0 + y_0\rVert$ with $y_0\in \Phi_{cm}$, see also Fig.~\ref{distance_near}. To further simplify the analysis, we neglect this correlation and assume that the serving distances are \ac{i.i.d.} Rayleigh distributed with marginal distributions $f_{H_i}(h_i)=\mathrm{Rayleigh}(h_i,\sqrt{2}\sigma)$. Hence, the conditional \ac{PDF} of the joint serving distances' distribution $f_{\boldsymbol{H}_k}(\boldsymbol{h}_k)$ is directly obtained from
\begin{align}
\label{joint-pdf} 
f_{\boldsymbol{H}_k}(\boldsymbol{h}_k) = \prod_{i=1}^{k}\frac{h_i}{2\sigma^2}e^{-\frac{h_i^2}{4\sigma^2}}
\end{align}
Conditioning on having $k$ catering devices, i.e., $s_{\Phi_{cm}}=\sum_{i=1}^{k}  h_i^{-\alpha}$, now, the probability $\mathbb{P}[R_{m}>\rho]$ in (\ref{eq:offload}) can be written as
\begin{align}
\label{given-k} 
\mathbb{P}[R_{m}>\rho|k] = \int_{\boldsymbol{h}_k=\boldsymbol{0}}^{\infty} \mathscr{L}_{I} \Bigg(\frac{\theta}{\gamma_d\sum_{i=1}^{k}  h_i^{-\alpha}}
\Bigg|k\Bigg) f_{\boldsymbol{H}_k}(\boldsymbol{h}_k) \dd{\boldsymbol{h}}_k
\end{align}
Given that $\Phi_{cm}$ is a \ac{PPP} , the number of catering devices for content $m$ is a Poisson \ac{RV} with mean $c_m\overline{n}$. Therefore, the probability that there are $k$ catering devices is equal to $\frac{(\overline{n}c_m)^ke^{-c_m\overline{n}}}{k!}$. Invoking this along with (\ref{eq:lap-transform}), (\ref{joint-pdf}), and (\ref{given-k}) into (\ref{eq:offload}), $\mathbb{P}_{o}(\boldsymbol{c})$ is given at the top of next page in (\ref{ch4:offloading-gain-eqn}). 

Since the obtained expression for $\mathbb{P}_{o}(\boldsymbol{c})$ in (\ref{ch4:offloading-gain-eqn}) involves multi-fold numerical integrals and summations, this renders our offloading maximization problem intractable. To circumvent this difficulty and gain a more general understanding of the problem at hand, we derive a lower bound and an approximation on $\mathbb{P}_{o}(\boldsymbol{c})$ in the sequel.
\newcounter{MYtempeqncnt}
\begin{figure*}[!t]
\normalsize
\setcounter{MYtempeqncnt}{\value{equation}}
\begin{align}
\label{ch4:offloading-gain-eqn} 
\mathbb{P}_{o}(\boldsymbol{c}) =  \sum_{m=1}^{N_f} q_m\Bigg(c_m + \big(1-c_m\big).
\underbrace{\sum_{k=1}^{\infty} \frac{(\overline{n}c_m)^ke^{-c_m\overline{n}}}{k!}
 \int_{\boldsymbol{h}_k=\boldsymbol{0}}^{\infty} {\rm exp}\Big(-2\pi \lambda_p\int_{v=0}^{\infty}\Big(1 - e^{-\overline{n}(1 - \zeta(v,t))}\Big)v\dd{v}\Big) \prod_{i=1}^{k}\frac{h_i}{2\sigma^2}e^{-\frac{h_i^2}{4\sigma^2}} \dd{\boldsymbol{h}}_k}_{\mathbb{P}(R_{m}>\rho)} \Bigg)
\end{align}
\setcounter{MYtempeqncnt}{\value{equation}}
\begin{align}
\label{ch4:lower-bound-offload-gain} 
\mathbb{P}_{o}^{\sim}(\boldsymbol{c}) &= \sum_{m=1}^{N_f} q_m\Big(c_m + \big(1-c_m\big).
\underbrace{\sum_{k=1}^{\infty} \frac{(\overline{n}c_m)^ke^{-c_m\overline{n}}}{k!}
\int_{\boldsymbol{h}_k=\boldsymbol{0}}^{\infty} e^{-\pi \overline{n}\lambda_p (\frac{\theta}{\sum_{i=1}^{k} h_i^{-\alpha}})^{2/\alpha} \Gamma(1 + 2/\alpha)\Gamma(1 - 2/\alpha)}
 \prod_{i=1}^{k}\frac{h_i}{2\sigma^2}e^{-\frac{h_i^2}{4\sigma^2}} \dd{\boldsymbol{h}}_k}_{\mathbb{P}(R_{m}>\rho)}\Big)	\end{align}
\setcounter{equation}{\value{MYtempeqncnt}+1}
\hrulefill
\end{figure*}
\begin{proposition}
\label{ch4:comp-interference-approx}	
The Laplace transform of interference derived in \Eq{lap-transform} can be bounded by
 \begin{equation} 
 \label{eq:ppp-interference}
\Lc_{I}(t) \approx \exp\left(-\pi \overline{n}\lambda_p t^{2/\alpha} \Gamma(1 + 2/\alpha)\Gamma(1 - 2/\alpha)\right).
\end{equation}
\end{proposition}
\begin{IEEEproof}
Please refer to \App{proof-comp-interference-approx}. 
\end{IEEEproof}

\begin{remark} {\rm 
The expression for the bound on the conditional Laplace transform of the interference in (\ref{eq:ppp-interference}) is identical to the (un-conditional) Laplace transform of a \ac{PPP} with density $\overline{n}\lambda_p$, which means that the obtained approximation is, indeed, a lower bound on the Laplace transform, accuracy of which we show in \Fig{lower-bound}.}
\end{remark}

Plugging this result into (\ref{eq:offload}), a lower bound on the offloading gain, denoted as $\mathbb{P}_{o}^{\sim}(\boldsymbol{c})$, is given in (\ref{ch4:lower-bound-offload-gain}) at the top of next page.  
\section{Optimized Caching Probabilities}
Although $\mathbb{P}_{o}^{\sim}(\boldsymbol{c})$ is much simpler to compute as compared to $\mathbb{P}_{o}(\boldsymbol{c})$, it is still challenging to obtain the optimal caching probability due to the multi-fold integration in (\ref{ch4:lower-bound-offload-gain}). In what follows, we consider a special case for $\mathbb{P}_{o}^{\sim}(\boldsymbol{c})$ wherein the offloading gain maximization problem turns out to be convex. 
\subsubsection{One Serving Device (k=1)}		
In this case, we solve for the caching probability for $k=1$ and $t = \frac{\theta h^{\alpha}}{\gamma_d}$. Starting from (\ref{ch4:lower-bound-offload-gain}) at the top of next page with $k=1$, we have 
\begin{align}
\mathbb{P}(R_{m}>\rho) &= (\overline{n}c_m)e^{-c_m\overline{n}}\cdot
\nonumber \\
&\int_{h=0}^{\infty} e^{-\pi \overline{n}\lambda_p (\theta h^{\alpha})^{2/\alpha} \Gamma(1 + 2/\alpha)\Gamma(1 - 2/\alpha)}
 \frac{h}{2\sigma^2}e^{-\frac{h^2}{4\sigma^2}} \dd{h}
\end{align} 
Solving the integral in the above, and substituting in (\ref{ch4:lower-bound-offload-gain}), we get
$\mathbb{P}_{o}^{\sim1}(\boldsymbol{c})$ written as 
\begin{equation}
\label{ch4:lower-bound-closed-form}
\mathbb{P}_{o}^{\sim1}(\boldsymbol{c})\overset{}{=}\sum_{m=1}^{N_f} q_m\Big(c_m + \big(1-c_m\big) (c_m\overline{n})e^{-c_m\overline{n}} 
   \frac{1}{\mathcal{Z}(\theta,\alpha,\sigma) }\Big)
\end{equation}   
where $\mathcal{Z}(\theta,\alpha,\sigma) = 4\sigma^2\pi\overline{n}\lambda_p \theta^{2/\alpha}\Gamma(1 + 2/\alpha)\Gamma(1 - 2/\alpha)+ 1$.

Hence, the desirable caching placement can be found by solving the following problem
\begin{align}
\label{optimize_eqn_p}
\textbf{P1:}		\quad &\underset{\boldsymbol{c}}{\text{max}} \quad \mathbb{P}_{o}^{\sim1}(\boldsymbol{c}) \\
\label{const110}
&\textrm{s.t.}\quad  \sum_{n=1}^{N_f} c_m = M, \quad   c_m \in [ 0, 1] 
\end{align}

The optimal solution for \textbf{P1} is formulated in the following Lemma.	
\begin{lemma}
\label{ch4:offload-concave}
The lower bound on the offloading gain $\mathbb{P}_{o}^{\sim1}(\boldsymbol{c})$ in (\ref{ch4:lower-bound-closed-form}) is concave w.r.t.  the caching probability, and the optimal probabilistic caching $\underline{\boldsymbol{c}}^*$ for \textbf{P1} is given by
      \[
    \underline{c_{m}}^{*}=\left\{
                \begin{array}{ll}
                  1 \quad\quad\quad , v^{*} <   q_m - \frac{q_m \overline{n}e^{-\overline{n}}}{\mathcal{Z}}\\ 
                  0   \quad\quad\quad, v^{*} >   q_m + \frac{q_m \overline{n}}{\mathcal{Z}}\\
                 \psi(v^{*}) \quad, {\rm otherwise}
                \end{array}   
              \right.
  \]
where $\psi(v^{*})$ is the solution of $v^{*} =   q_m + \frac{q_m \overline{n}e^{-\underline{c_{m}}^{*} \overline{n}}}{\mathcal{Z}}
\big(1 - \underline{c_{m}}^{*}(2 +\overline{n} - \overline{n}\underline{c_{m}}^{*}) \big)$ that satisfies $\sum_{m=1}^{N_f} \underline{c_{m}}^{*}=M$.
\end{lemma}
\begin{proof}
The details are omitted due to the limited space.
\end{proof}
It is worth mentioning that the optimal caching solution $\underline{\boldsymbol{c}}^*$ for \textbf{P1} is strictly suboptimal relative to the caching solution for a problem with $k$ cooperative caching devices, where $k$ is random. 
However, when substituted to (\ref{ch4:offloading-gain-eqn}), it provides insights into system design and allows us to quantify the performance improvements over traditional caching techniques.
\section{Numerical Results}
\begin{table}[ht]
\caption{Simulation Parameters} 
\centering 
\begin{tabular}{c c  c} 
\hline\hline 
Description & Parameter & Value  \\ [0.5ex] 
\hline 
Displacement standard deviation & $\sigma$ & \SI{50}{\metre} \\ 
Popularity index&$\beta$&0.5\\			
Path loss exponent&$\alpha$&4\\
Library size&$N_f$&100 files\\
Cache size per device&$M$&5 files\\
Mean number of devices per cluster&$\overline{n}$&8\\
Density of clusters&$\lambda_{p}$&40 clusters/\SI{}{km}$^2$ \\
$\sir$ threshold&$\theta$&\SI{0}{\deci\bel}\\
\hline 
\end{tabular}
\label{ch4:table:sim-parameter} 
\end{table}
At first, we validate the developed mathematical model via Monte Carlo simulations. Then we benchmark the proposed \ac{PC} against conventional caching schemes. Unless otherwise stated, the network parameters are selected as shown in Table \ref{ch4:table:sim-parameter}. 
\begin{figure}[tb]	
\centering
\includegraphics[width=0.35\textwidth]{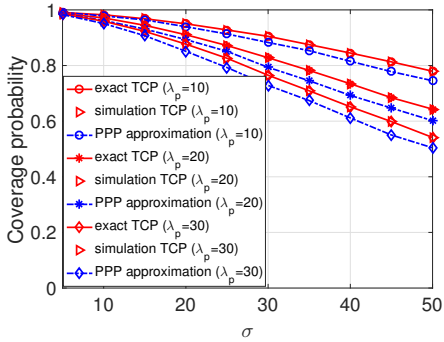}		
\caption {\textcolor{black}{The derived lower bound on $\mathbb{P}(R_{m}>\rho)$ in (\ref{ch4:lower-bound-offload-gain}) versus intra-cluster displacement variance $\sigma$ is plotted for various parent \ac{PPP} density $\lambda_p$, $c_m=1$. "exact TCP" in the legend refers to the exact expression obtained for the \ac{TCP} while "PPP approximation" refers to the lower bound of (\ref{eq:ppp-interference}).}}
\label{fig:lower-bound}
\vspace{-0.4cm}
\end{figure}  
We hereafter call the coverage probability of content $m$ and $\mathbb{P}(R_{m}>\rho)$ interchangeably. 
In Fig.~\ref{fig:lower-bound}, we plot the closed-form expression, simulation, and lower bound on the coverage probability of content $m$ versus displacement variance $\sigma$ for various parent \ac{PPP} densities $\lambda_p$. The theoretical and simulated results for the coverage probability are consistent. The derived lower bound is considerably tight when both $\sigma$ and $\lambda_p$ are relatively small. Also, it is noticeable that the coverage probability $\mathbb{P}(R_{m}>\rho)$ monotonically decreases with both $\sigma$ and $\lambda_p$, which reflects the fact that the desired signal is weaker when the distance between catering devices and the receiver is larger, and the effect of inter-cluster interference increases when the density of clusters increases, respectively. When $\lambda_p$ and $\sigma$ increase, the obtained lower bound becomes no longer tight, however, it still represents a decent tractable bound on the exact coverage probability.

\begin{figure}[tb]	
\centering
\includegraphics[width=0.35\textwidth]{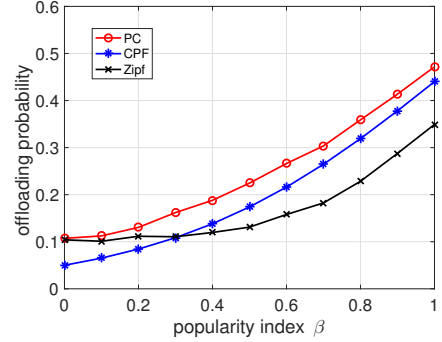}		
\caption {The offloading probability versus the popularity of files $\beta$ under different caching schemes, namely, PC, Zipf, and CPF.}
\label{offloading_gain_vs_beta}
\vspace{-0.4cm}
\end{figure}
Fig.~\ref{offloading_gain_vs_beta} manifests the prominent effect of the files' popularity on the offloading gain. We compare the offloading gain of three different caching schemes, namely, the proposed \ac{PC}, Zipf's caching (Zipf), and \ac{CPF}. We can see that the offloading gain under the \ac{PC} scheme attains the best performance as compared to other schemes. Also, we note that both \ac{PC} and Zipf schemes encompass the same offloading gain when $\beta=0$ owing to the uniformity of content popularity.

\begin{figure}
    \centering
    \subfigure[Case one]
    {
        \includegraphics[width=1.6in]{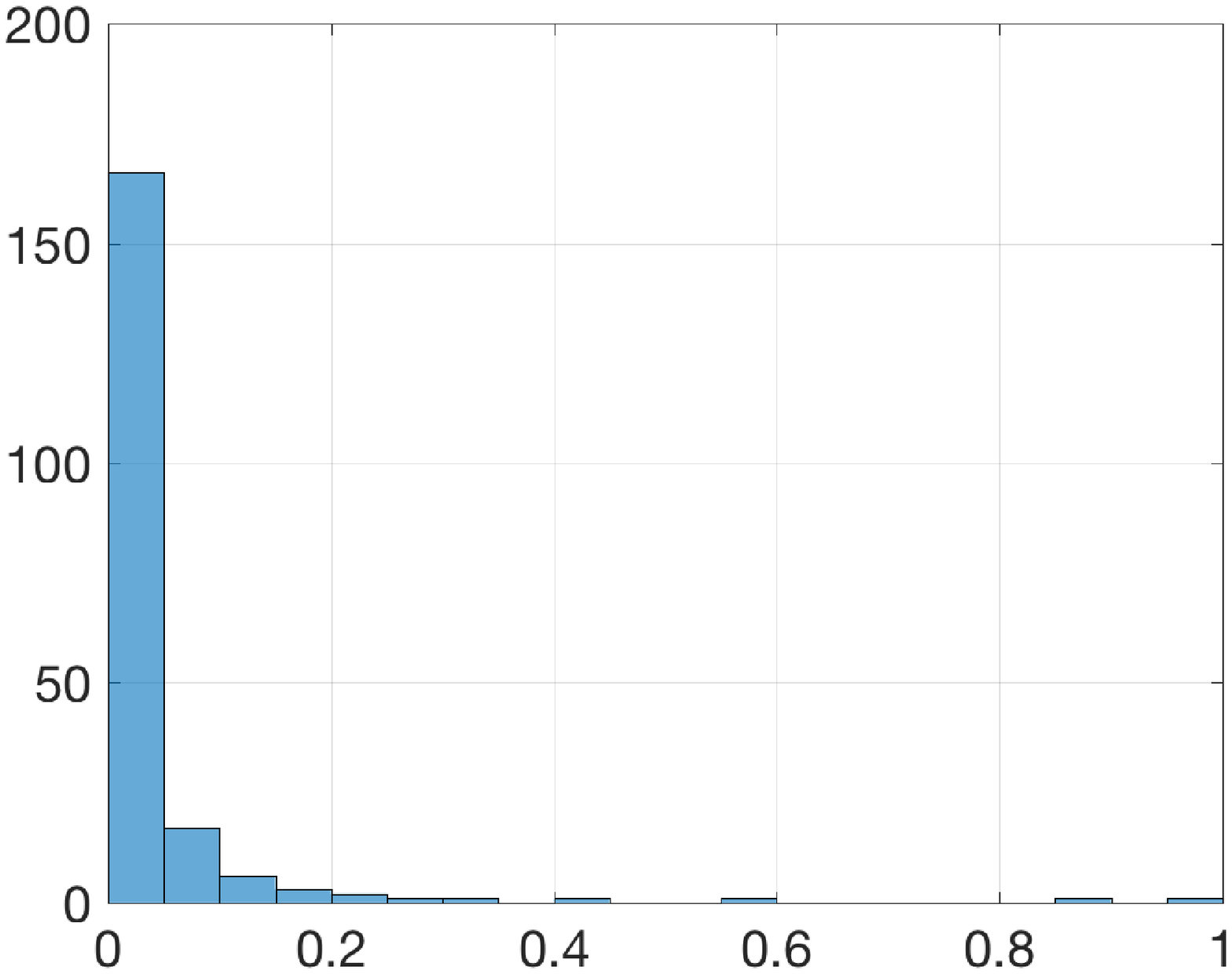}		
        \label{histogram_b_i_p_star}
    }
    \subfigure[Case two]
    {
        \includegraphics[width=1.6in]{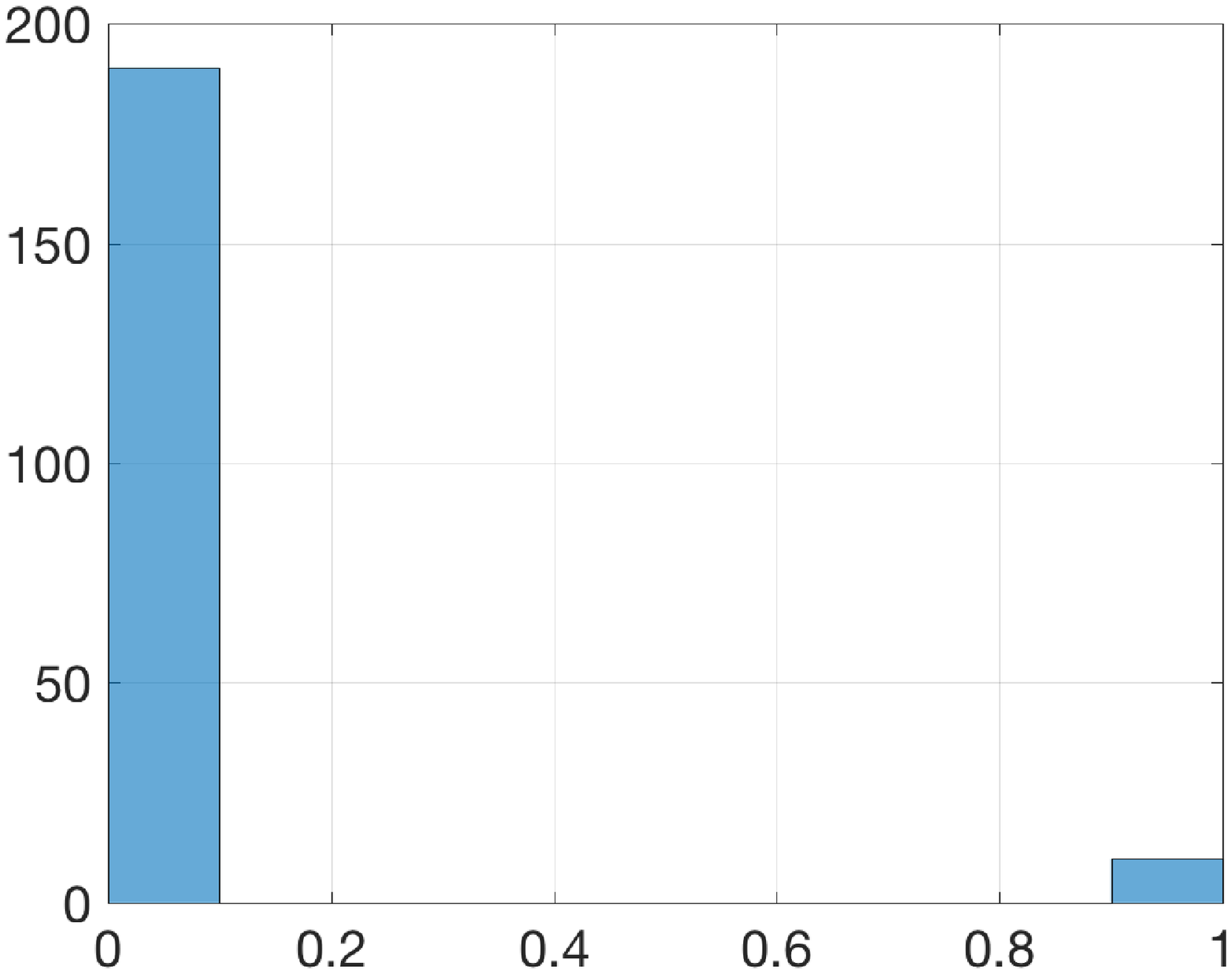}
        \label{histogram_b_i_p_leq_p_star}
    }
    \caption{Histogram of the caching solution $\underline{\boldsymbol{c}}^*$ when (a) $\sigma=\SI{10}{m}$ and $\lambda_p=20$ clusters/\SI{}{km}$^2$ and (b) $\sigma=\SI{100}{m}$ and $\lambda_p=50$ clusters/\SI{}{km}$^2$.}		
    \label{histogram_b_i}
    \vspace{-0.4cm}
\end{figure}
To show the trade-off between content diversity gain and cooperative transmission gain, in Fig.~\ref{histogram_b_i}, we plot the histogram of the caching solution $\underline{\boldsymbol{c}}^*$ for two cases. 
The first case, in Fig.~\ref{histogram_b_i_p_star}, represents a system with good transmission conditions, i.e., both $\sigma$ and $\lambda_p$ are small. The second case in Fig.~\ref{histogram_b_i_p_leq_p_star}, on contrary, is for large values of $\sigma$ and $\lambda_p$. It is clear from the histograms that the caching solution $\underline{\boldsymbol{c}}^*$ tends to be more uniform and
skewed for cases one and two, respectively. 
\textcolor{black}{This shows that the caching probability tends to achieve diverse contents cached among the devices for a system with good transmission conditions. However,  for a system with poor transmission conditions, caching redundancy is favorable to combat such adverse conditions via cooperative transmission.} 
\section{Conclusion}
In this work, we propose a cooperative transmission scheme and probabilistic caching for a clustered device-to-device (D2D) network. We first derive a closed-form expression for the offloading gain and then obtain a simpler yet tight lower bound on the offloading gain based on an upper bound on the interference power. \textcolor{black}{Considering a special case when downloading content from only one serving device,  
the offloading gain maximization problem is shown to be convex, and the caching probability maximizing  the simplified offloading gain is then obtained. We show that the network scaling parameters, e.g., $\sigma$ and $\lambda_p$, control the trade-off between content diversity gain and cooperative transmission gain. Results show up to $12\%$ increase in the offloading gain compared to the Zipf's caching technique. In the extended version of this paper \cite{amer2019joint}, we obtain an approximated yet accurate expression of the offloading gain that allows us to find the optimized caching probability.}
\begin{appendices}
\section{Proof of Lemma \ref{ch4:comp-interference}}				
\label{app:proof-comp-interference}
 \textcolor{black}{In the following, by saying $u \in \Phi_{c}$, we mean that $y \in \Phi_{c}$ where $u=\lVert x+y\rVert$.}
\begin{align}
 &\mathscr{L}_{I}(t) = \mathbb{E} \Bigg[e^{-t\gamma_d \sum_{\Phi_p^{!}} \sum_{u \in \Phi_{c}}  G_{u}  u^{-\alpha}} \Bigg] \nonumber \\ 
\label{ch4:interference-back}
  &= \mathbb{E}_{\Phi_p} \Big[\prod_{\Phi_p^{!}} \mathbb{E}_{\Phi_{c}} \prod_{u \in \Phi_{c}} \mathbb{E}_{u,G_{u}}  e^{-t\gamma_d G_{u}  u^{-\alpha}} \Big]  \\
    &\overset{(a)}{=} \mathbb{E}_{\Phi_p} \Big[\prod_{\Phi_p^{!}} \mathbb{E}_{\Phi_{c}} \prod_{u \in \Phi_{c}} \mathbb{E}_{u}\frac{1}{1+t\gamma_d u^{-\alpha}} \Big] \nonumber \\
    \label{LT_c0}
        &\overset{(b)}{=} \mathbb{E}_{\Phi_p} \Big[\prod_{\Phi_p^{!}} \sum_{l=0}^{\infty}  \Big(
        \int_{u=0}^{\infty}\frac{1}{1+t\gamma_d u^{-\alpha}} f_{U|V}(u|v)\dd{u}\Big)^{l} P(n=l) \Big] 
        \nonumber \\
        &\overset{(c)}{=} {\rm exp}\Big(-2\pi \lambda_p\int_{v=0}^{\infty}\Big(1 - \sum_{l=0}^{\infty} \nonumber \\
 &\Big(\int_{u=0}^{\infty}\frac{1}{1+t\gamma_d u^{-\alpha}} f_{U}(u|v)\dd{u}\Big)^{l} \frac{\overline{n}^le^{-n}}{l!}\Big)v\dd{v}\Big)
\end{align}   
where $G_{u}=G_{y}$ for ease of notation. (a) follows from the Rayleigh fading assumption, (b) follows from the \ac{TCP} definition with the number of devices per cluster being a Poisson \ac{RV}, and $P(n=l)= \frac{\overline{n}^le^{-\overline{n}}}{l!}$ is the probability that there are $l$ catering devices caching content $m$ per each remote cluster, with $c_m=1$ for the assumed worst case interference scenario. (c) follows from the \ac{PGFL} for a \ac{PPP}. We denote $\int_{u=0}^{\infty}\frac{1}{1+t\gamma_d u^{-\alpha}} f_{U|V}(u|v)\dd{u}$ as $\delta(v,t)$. Using the fact that (Taylor's expansion of) $e^x=\sum_{n=0}^{\infty}\frac{x^n}{n!}$, we rearrange $\mathscr{L}_{I}(t) $ in the above equation as  
\begin{align}
 \label{LT_fpgl0}
\mathscr{L}_{I}(t)&=  {\rm exp}\Big( -2\pi \lambda_p\int_{v=0}^{\infty}\big(1 - {\rm exp}\big(-\overline{n}(1 - \delta(v,t))\big)\big)v\dd{v}\Big)
\end{align}
By denoting $1 - \delta(v,t)= 1 -\int_{u=0}^{\infty}\frac{1}{1+t\gamma_d u^{-\alpha}} f_{U|V}(u|v)\dd{u}= \int_{u=0}^{\infty}\frac{t\gamma_d}{u^{\alpha}+t\gamma_d} f_{U|V}(u|v)\dd{u}$ as $\zeta(v,t)$, Lemma \ref{ch4:comp-interference} is proven.

\section{Proof of Proposition \ref{ch4:comp-interference-approx}}
\label{app:proof-comp-interference-approx}
By conditioning on $S_{\Phi_{cm}}=s_m=\sum_{i=1}^{k}  h_i^{-\alpha}$, we derive a bound on the Laplace transform of inter-cluster interference based on Taylor's series expansion. Starting from equation (\ref{ch4:interference-back}) in \App{proof-comp-interference}, we have 		
\begin{align}
 \Lc_{I}(t|k)&= \mathbb{E}_{\Phi_p} \Bigg[\prod_{\Phi_p^{!}} \mathbb{E}_{\Phi_{c}} \prod_{u \in \Phi_{c}} \mathbb{E}_{u,G_{u}}  {\rm exp}\big(-t G_{u}  u^{-\alpha}\Big) \Bigg] \nonumber 	\\ 
 &\overset{(a)}{=}\Eb_{G_{u}} e^{-2\pi \lambda_p\int_{v=0}^{\infty}\big(1 - e^{-\overline{n}(1 - \zeta'(v,t))}\big)v\dd{v}}		\nonumber	 \\ 
 &\overset{(b)}{\approx} \Eb_{G_{u}} e^{-2\pi \lambda_p\int_{v=0}^{\infty}\big(1 -  (1 - \overline{n}(1 - \zeta'(v,t))\big)v\dd{v}}		\nonumber		\\ 
    &\overset{}{=}\Eb_{G_{u}} e^{-2\pi \lambda_p \overline{n} \int_{v=0}^{\infty}(1 - \zeta'(v,t))v\dd{v}}\nonumber	\\
\label{ch4:G-of-s}
&=e^{-2\pi\overline{n}\lambda_p \big(\int_{v=0}^{\infty}v\dd{v} - \Eb_{G_{u}} \int_{v=0}^{\infty}\zeta'(v,t) v\dd{v}\big)} 
\end{align}
where  $\zeta'(v,t) = \int_{u=0}^{\infty}e^{-t\gamma_dG_{u}u^{-\alpha}} f_{U|V}(u|v)\dd{u}$ and $G_{u}= G_{y}$; (a) follows from tracking the proof of Lemma \ref{ch4:comp-interference} up until equation (\ref{LT_fpgl0}) and (b) follows from Taylor series expansion for exponential function $e^{-x} = 1 - x - \Theta(x^2)$ when $x$ is small. It is worth mentioning that the obtained $\mathscr{L}_{I}(t|k)$ in (\ref{ch4:G-of-s}) is the Laplace transform of an upper bound on the interference. Correspondingly, the resulting $\mathbb{P}(R_{m}>\rho)$ and $\mathbb{P}_{o}^{\sim}(\boldsymbol{c})$ are lower bounds on their exact values. We proved in \cite[Lemma 2]{amer2018minimizing} that $\int_{v=0}^{\infty}v\dd{v} - \Eb_{G_{u}} \int_{v=0}^{\infty}\zeta'(v,t) v\dd{v}= \frac{(t\gamma_d)^{2/\alpha}}{2} \Gamma(1 + 2/\alpha)\Gamma(1 - 2/\alpha)$, hence, Proposition \ref{ch4:comp-interference-approx} is proven.
\end{appendices}

\bibliographystyle{IEEEtran}
\bibliography{bibliography}
\end{document}